\documentclass[aps,pra,twocolumn,superscriptaddress,amsmath,amssymb,showpacs,notitlepage,longbibliography]{revtex4-1}
\pdfoutput=1
\usepackage{algorithm}
\usepackage{color}
\usepackage{algorithmic}
\usepackage{amsmath}
\usepackage{braket}
\usepackage{multirow}
\usepackage{amssymb}
\usepackage{amsthm}
\usepackage{bm}
\usepackage{thmtools}

\usepackage{array}
\usepackage{url}
\usepackage{hyperref}
\hypersetup{
    colorlinks=true,
    linkcolor=blue,
    filecolor=magenta,      
    urlcolor=cyan,
    citecolor=red
}
\usepackage{graphicx}
\usepackage{bibentry}
\usepackage{booktabs}
\usepackage{dcolumn}
\usepackage{bm}
\usepackage{blkarray}
\usepackage{xparse}
\usepackage{mathtools}
\usepackage{microtype}
\usepackage{enumerate}
\usepackage{soul}

\usepackage{makecell}

\providecommand{\U}[1]{\protect\rule{.1in}{.1in}}
\newtheorem*{theorem*}{Theorem}

\newtheorem{conjecture}{Conjecture}

\newtheorem{proposition}{Proposition}
\newtheorem{remark}{Remark}

\begin{document}

\title{Upper bounds on the purity of Wigner non-negative quantum states that verify the Wigner entropy conjecture}

\author{Qipeng Qian}
\affiliation{Program in Applied Mathematics, The University of Arizona, Tucson, Arizona 85721, USA}
\author{Christos N. Gagatsos}
\affiliation{Department of Electrical and Computer Engineering, The University of Arizona, Tucson, Arizona, 85721, USA}
\affiliation{Wyant College of Optical Sciences, The University of Arizona, Tucson, Arizona, 85721, USA}
\affiliation{Program in Applied Mathematics, The University of Arizona, Tucson, Arizona 85721, USA}

\begin{abstract}
We present analytical results toward the Wigner entropy conjecture, which posits that among all physical Wigner non-negative states the Wigner entropy is minimized by pure Gaussian states for which it attains the value $1+\ln\pi$. 
Working under a minimal set of constraints on the Wigner function, namely, non-negativity, normalization, and the pointwise bound $\pi W\le 1$, we construct an explicit hierarchy of lower bounds $B_n$ on $S[W]$ by combining a truncated series lower bound for $-\ln x$ with moment identities of the Wigner function. 
This yields closed-form sufficient conditions, expressed in terms of the state purity $\mu:=\operatorname{Tr}(\rho^2)=2\pi\int dq\,dp\,W(q,p)^2$, ensuring $S[W]\ge 1+\ln\pi$. 
In particular, we first prove that all physical Wigner-non-negative states with $\mu\le 4-2\sqrt{3}$ satisfy the Wigner entropy conjecture.
We further obtain a systematic purity-only relaxation of the hierarchy,
whose limiting sufficient condition is $\mu\le 2/e$. 
Finally, we show that the threshold $2/e$ is sharp under the relaxed
constraints considered here, thereby identifying the need for additional
quantum-realizability information in the remaining high-purity regime. 
\end{abstract}
\maketitle

\section{Introduction}\label{sec:intro}

Uncertainty relations lie at the heart of quantum theory, quantifying the incompatibility of non-commuting observables and shaping fundamental limits in sensing, communication, and information processing. It originates from Heisenberg's principle \cite{heisenberg1927anschaulichen}, which quantifies the incompatibility of the quadratures $q$ and $p$ via variances. A widely used refinement replaces variances by entropies; see, e.g., \cite{Coles2017}. In particular, Bia{\l}ynicki-Birula and Mycielski derived an entropic uncertainty relation that lower bounds the sum of the Shannon differential entropies of the position and momentum densities, and is stronger than the variance-based bound \cite{bialynicki1975uncertainty}. 

Moreover, since Shannon differential entropy is subadditive, a lower bound on the entropy of a joint phase-space density would imply (and potentially strengthen) the position--momentum entropic uncertainty relation. 
A closely related quantity is the Wehrl entropy, defined as the Shannon differential entropy of the Husimi $Q$ function. Wehrl conjectured that it is minimized by coherent states \cite{Wehrl1979}; this was proved by Lieb \cite{Lieb1978} and later extended, including sharp results for R\'enyi--Wehrl entropies in broad coherent-state settings \cite{LiebSolovej2014}. 
In this sense, a Wehrl-type lower bound refers to an entropy minimization principle on phase space. The Wigner entropy conjecture below is analogous in spirit, but it concerns the Shannon entropy of a non-negative Wigner function instead of the Husimi $Q$ function. Note that the latter is not invariant under symplectic transformations \cite{hertzthesis,Hertz_2019}, i.e., the entropy of the Husimi $Q$ function for the vacuum and single-mode squeezed vacuum states, is not the same. 

In continuous-variable quantum optics, the Wigner function $W(q,p)$ is arguably the most central phase-space representation. It is a real quasi-probability distribution which describes fully the quantum state, but it can take negative values. Of special interest is the convex subset of Wigner-non-negative states, whose Wigner function is everywhere non-negative and can therefore be viewed as a genuine probability density on phase space.  Motivated by this classical-looking interpretation and the subadditivity of Shannon’s differential entropy, \cite{VanHerstraeten2021quantum,hertz2017,hertzthesis} introduced the Wigner entropy as the Shannon differential entropy of the Wigner function $W(q,p)$ and conjectured a Wehrl-type lower bound: 
\begin{conjecture}
\label{Conj1}
For any Wigner non-negative state, 
\begin{eqnarray}
    \label{eq: wigner conj}
    S[W]\geq 1+\ln\pi,    
\end{eqnarray}
while $S$ denotes for Shannon entropy and the lower bound is attained by any pure Gaussian state.
\end{conjecture}

Specifically, the marginals of the Wigner function coincide with the probability densities of the variables $q$ and $p$, defined by $P_q=\int dp W(q,p)$ and $P_p=\int dq W(q,p)$. The Bia{\l}ynicki-Birula-Mycielski inequality \cite{bialynicki1975uncertainty}, together with the subadditivity of Shannon’s differential entropy, implies 
\begin{eqnarray}
\label{Bialynicki}
S[P_q]+S[P_p]&\geq& 1+\ln\pi,\\
\label{eq:subadd} S[P_q]+S[P_p]&\geq& S[W], 
\end{eqnarray}
If Conjecture \ref{Conj1} holds, then inequalities \eqref{Bialynicki} and \eqref{eq:subadd} can be combined as 
\begin{eqnarray}
S[P_q]+S[P_p]\geq S[W] \geq 1+\ln \pi,
\end{eqnarray}
which yields a stronger entropic uncertainty relation for Wigner non-negative states. 

Here and throughout, we denote the state purity by
$\mu:=2\pi\int dq\,dp\,W^2(q,p)$.
The pure-state endpoint is characterized by Hudson's theorem
\cite{Hudson1974}, which states that a pure continuous-variable state has
an everywhere non-negative Wigner function if and only if it is Gaussian.
Consequently, the physical Wigner-non-negative states with $\mu=1$ are
precisely pure Gaussian states and attain equality in
Conjecture~\ref{Conj1}. 

A convenient one-parameter family that interpolates between the Shannon case and higher moments is given by the Wigner--R\'enyi entropies defined in \cite{Kalka2023,Dias2023}, which reduce to the Wigner entropy when $\alpha\to 1$. The Wigner entropy and phase-space methods have been used in diverse contexts, ranging from noisy polarizers \cite{Aifer2023} to high-energy physics \cite{Barata2023} and non-equilibrium field theory \cite{Reiche2022}; see also \cite{Garttner2023,VanHerstraeten2023continuous,Salazar2023} for broader phase-space perspectives. 

While the conjecture remains unanswered for the general Wigner-non-negative states, several partial results have been proven: \cite{Dias2023} proved that the conjectured lower bound holds for the Wigner--R\'enyi entropy of all Wigner-non-negative states in the regime for R\'enyi parameter $\alpha\ge 2$, but their argument does not extend to $\alpha \rightarrow 1$, which is what a full proof of the Wigner entropy conjecture would require. Through a different direction, \cite{LiebSolovej2014} established sharp coherent-state entropy inequalities for all R\'enyi parameters via the Wehrl-to-Wigner connection for beam-splitter outputs with one vacuum input. This implies the conjectured bound for a special beam-splitter subclass for all $\alpha$ (including $\alpha \rightarrow 1$). Furthermore, \cite{VanHerstraetenCerf2025} recently proved the conjecture for the Wigner--R\'enyi entropy for R\'enyi parameter $\alpha\ge 1/2$, which in particular includes the Wigner entropy but is for beam-splitter states only. 

In addition, more specific Wigner-entropy results have been obtained.  The conjectured Wigner entropy lower bound holds for all Wigner non-negative passive states (mixtures of Fock states with non-increasing weights) is provided in \cite{VanHerstraeten2021quantum}. Using an underlying continuous majorization relation, \cite{VanHerstraeten2023continuous} proved the bound for any Wigner non-negative
mixture of the first three Fock states. Finally, in \cite{PhysRevA.110.012228} the validity of the conjecture was proven for (generally mixed states) written on the Fock states $\{|0\rangle,|1\rangle\}$ basis and a sufficient condition for the conjecture's validity for general mixed states was derived. 

Taken together, the above results establish the conjectured lower bound by exploiting additional structure. Our result enters from a different direction. We do not try to sharpen these class-specific proofs; instead, we ask what can already be certified from the universal constraints shared by every single-mode Wigner-non-negative state, namely positivity, normalization, and the pointwise Wigner bound. Thus the present approach is complementary to the existing literature. Starting from these minimal constraints also provides a useful baseline for understanding what information is missing from a purely universal argument. In this sense, the gap between our sufficient bounds and the conjectured bound is not merely a limitation, but also indicates where additional quantum or state-specific structure must enter. 

\medskip
\noindent
\textbf{Our contribution.}
In this paper, we develop a complementary, minimal-constraint approach aimed at producing simple purity-based sufficient conditions for the Wigner entropy conjecture. Rather than exploiting a detailed characterization of the set of Wigner-nonnegative quantum states, we assume only a small collection of constraints that are necessary for any single-mode Wigner-non-negative state: 
\begin{enumerate}
\item[(i)] $W(q,p)\ge 0$ and $\int dq\,dp\, W(q,p)=1$; 
\item[(ii)] the pointwise bound $\|W\|_\infty \le 1/\pi$. 
\end{enumerate}
Under these constraints, we derive an explicit, analytically tractable family of lower bounds on the Wigner entropy $S[W]$. Our approach yields a systematic hierarchy $B_n$ and, in particular, simple purity-based sufficient conditions ensuring the conjectured bound $S[W]\ge 1+\ln\pi$. 
This culminates in the closed-form threshold $\mu\le 2/e$, which provides a fast check for broad families of Wigner-non-negative states. Together with the pure-state endpoint $\mu=1$ characterized by Hudson's theorem \cite{Hudson1974}, this leaves the intermediate mixed-state interval $2/e<\mu<1$ outside the scope of our general purity-only certificate. 

To the best of our knowledge, this work provides the first explicit and systematic purity-based route toward validating the Wigner-entropy conjecture under minimal constraints. Beyond the concrete thresholds reported here, the resulting moment-test viewpoint suggests a new direction for future improvements: once additional structural information is available (e.g., nontrivial constraints on higher moments or quantum-realizability conditions), the same framework can be strengthened in a controlled manner. In particular, we will show that, under constraints (i)--(ii) alone, the limiting purity-only condition $\mu\le 2/e$ is optimal in the sense that no sufficient condition depending only on $\mu$ can be improved without invoking further constraints. 

This paper is organized as follows.  In Sec.~\ref{sec: defs}, we fix notation and state the constraints we considered throughout this work. In Sec.~\ref{sec: bounds}, we derive our main purity-dependent bounds from a truncated series expansion of $-\ln x$. We conclude in Sec.~\ref{sec: conclusion} with comments and directions for further
improvements.

\section{Definitions and conditions}\label{sec: defs}

For clarity of notation, we focus on a single bosonic mode with canonical quadratures $(q,p)$ and associated Wigner function $W(q,p)$; the multi-mode generalization follows by replacing phase space $\mathbb{R}^2$ with $\mathbb{R}^{2n}$ and adjusting the corresponding $\pi$-normalizations, so we do not repeat it explicitly. 

Our object of interest is the Wigner entropy 
\begin{align}
S[W]:=-\int dq\,dp\, W(q,p)\,\ln W(q,p),
\label{eq: def wiger entropy}
\end{align}
together with its conjectured lower bound for Wigner-non-negative states, namely $S[W]\ge 1+\ln\pi$ [cf.~Eq.~\eqref{eq: wigner conj}]. 

In this work, rather than imposing the full set of quantum-realizability constraints that characterize physical Wigner functions, we adopt a minimal set of conditions under which explicit analytic bounds can be derived. 
Here, by quantum-realizability constraints we mean the conditions required for a phase-space function to arise as the Wigner function of a positive trace-one density operator. These constraints are stronger than ordinary positivity and normalization: a normalized non-negative function on phase space need not correspond to a physical quantum state. 
The conditions we assume are, 
\begin{eqnarray}
 \label{eq: Cond-nonneg}  W(q,p)&\geq& 0,\\
  \label{eq: Cond-regularization}  \int dq\,dp\, W(q,p)&=&1,\\
 \label{eq: Cond-pointwise}  \pi W(q,p) &\leq& 1.
\end{eqnarray}
Conditions \eqref{eq: Cond-nonneg} and \eqref{eq: Cond-regularization} allow us to treat $W$ as a normalized probability density.
The pointwise bound \eqref{eq: Cond-pointwise} has a direct quantum origin. In the single-mode case, the Wigner function can be represented, in the convention used here, as the expectation value of a displaced parity operator divided by $\pi$. Since the parity operator has eigenvalues $\pm 1$, its expectation value is bounded in absolute value, giving $|W(q,p)|\leq 1/\pi$ for any physical single-mode Wigner function. For Wigner-non-negative states, this reduces to $0\leq W(q,p)\leq 1/\pi$, which is precisely \eqref{eq: Cond-nonneg} together with \eqref{eq: Cond-pointwise}. 
The pointwise bound \eqref{eq: Cond-pointwise} is therefore a necessary constraint for physical Wigner functions, but it is not sufficient to guarantee that $W$ corresponds to a positive semidefinite density operator. Our results should therefore be understood as providing sufficient certificates for the conjectured entropy bound, rather than a characterization of all cases in which the conjecture holds: we derive bounds that hold for any function $W$ obeying \eqref{eq: Cond-nonneg}--\eqref{eq: Cond-pointwise}, and therefore also apply to any physical Wigner-non-negative state whenever the corresponding sufficient conditions are met. 

Equivalently, the admissible class considered here is a relaxation of the physical class: every Wigner-non-negative quantum state gives a function satisfying Eqs.~\eqref{eq: Cond-nonneg}--\eqref{eq: Cond-pointwise}, but the converse is false. 
Necessary and sufficient conditions for quantum realizability can be
formulated at the level of the symplectic Fourier transform of $W$ through the Kastler--Loupias--Miracle-Sol\'e (KLM), or quantum Bochner, conditions \cite{cordero2017positivity,NarcowichDistributions1989}. 
Roughly speaking, in addition to normalization and regularity, the
characteristic function must be of quantum positive type, meaning that an infinite family of phase-twisted matrices must be positive semidefinite for every finite collection of phase-space points. 
These conditions are necessary and sufficient for positivity of the associated Weyl operator, but they are substantially stronger than ordinary pointwise positivity of $W$ and are generally difficult to incorporate into explicit entropy estimates. An explicit function satisfying Eqs.~\eqref{eq: Cond-nonneg}--\eqref{eq: Cond-pointwise} but failing quantum realizability is given in Remark~\ref{rem:flat_nonphysical}. 

These constraints are deliberately chosen to allow us analytic calculations and to isolate a single, physically meaningful parameter. 
In particular, our bounds will be expressed in terms of the state purity, which in the Wigner representation is \cite{Hillery1984DistributionFunctions}: 
\begin{eqnarray}
\label{eq: mu}    
\mu = 2\pi \int dq\,dp\, W^2(q,p). 
\end{eqnarray}
For a general function satisfying only Eqs.~\eqref{eq: Cond-nonneg}--\eqref{eq: Cond-pointwise}, we use the same symbol $\mu$ for the quadratic functional in Eq.~\eqref{eq: mu}, even though it need not be the purity of a physical density operator. 
For physical states, the purity satisfies $\mu\in(0,1]$. Therefore, a proof of the full Wigner entropy conjecture by a purity-based route would require the conjectured lower bound to be established throughout this entire interval. The bounds we presented later in the paper do not achieve this: they reach the conjectured value only when $\mu\leq 2/e$. Thus the relaxation used here certifies the conjecture for a subset of Wigner-non-negative states, but it does not prove the remaining high-purity regime. This is precisely where the minimal constraints \eqref{eq: Cond-nonneg}--\eqref{eq: Cond-pointwise} become too coarse, and where additional quantum-realizability or state-specific constraints are expected to be necessary for approaching the full conjectured bound.

\section{Hierarchy of lower bounds and purity thresholds}\label{sec: bounds}

Having fixed the minimal constraints \eqref{eq: Cond-nonneg}--\eqref{eq: Cond-pointwise}, we now derive an explicit hierarchy of lower bounds on the Wigner entropy $S[W]$. 
The construction relies on a truncated series lower bound for $-\ln x$ on $x\in(0,1]$, which is applicable thanks to the pointwise constraint \eqref{eq: Cond-pointwise}.

\subsection{Truncated-series bound and the hierarchy $B_n$}\label{sec:derivation}

For $0<x\le 1$ and any integer $n\ge 1$, the truncated Taylor expansion of $-\ln x$ about $x=1$ yields the lower bound \cite{OlverNIST2010,RudinPMA} 
\begin{align}
    -\ln x \ge \sum_{k=1}^n \frac{(-1)^k}{k}(x-1)^k. 
    \label{eq: series}
\end{align}
Moreover, using the constraint \eqref{eq: Cond-regularization}, the Wigner entropy in Eq. \eqref{eq: def wiger entropy} can be rewritten as
\begin{align}
    \label{eq:ent2}
    S[W]=\ln\pi-\int dq\,dp\, W(q,p)\, \ln \big(\pi W(q,p)\big).
\end{align}
Under the constraint \eqref{eq: Cond-pointwise}, applying Eq. \eqref{eq: series} to Eq. \eqref{eq:ent2} yields
\begin{align}
\label{eq:lowerBoundn}
S[W]\ge& \ln\pi \nonumber\\
&+ \int dq\,dp\, W(q,p)\sum_{k=1}^n \frac{(-1)^k}{k}\big(\pi W(q,p)-1\big)^k \nonumber\\
=:& B_n,
\end{align}
which holds for all integers $n\ge 1$. The sequence $\{B_n\}_{n\ge 1}$ constitutes an analytically tractable hierarchy: larger $n$ incorporates higher-order moments of $W$ and will only improve the bound.

\subsection{Low-order examples: $n=1$ and $n=2$}\label{sec:examples}

For $n=1$, Eq. \eqref{eq:lowerBoundn} gives 
\begin{align}
    B_1 
    &= \ln \pi - \int dq\,dp\, W(q,p)\big(\pi W(q,p)-1\big) \nonumber\\
    &= 1+\ln\pi-\pi\int dq\,dp\, W^2(q,p) \nonumber\\
    &= 1+\ln\pi-\frac{\mu}{2},
\end{align}
where we used the constraint \eqref{eq: Cond-regularization} and the purity definition in Eq. \eqref{eq: mu}. 
Since $\mu>0$ for any nonzero $W$, enforcing $B_1\ge 1+\ln\pi$ would force $\mu\le 0$, and thus $n=1$ cannot yield a nontrivial sufficient condition.

For $n=2$, we obtain
\begin{align}
    \label{eq:B2_1}
    B_2 =& \ln \pi - \Big(\int dq\,dp\, W(q,p) \nonumber\\
    &\left[-(\pi W(q,p)-1)+\frac{1}{2}\left(\pi W(q,p)-1\right)^2\right]\Big), 
\end{align}
which can be simplified into
\begin{eqnarray}
\label{eq:B2_2}
B_2 = \frac{3}{2}+\ln\pi+\frac{\pi^2}{2}\int dq\,dp\, W^3(q,p) -\mu.
\end{eqnarray}
Since $W\ge 0$, the cubic term is nonnegative; dropping it yields the explicit relaxation
\begin{eqnarray}
B_2 \ge \frac{3}{2}+\ln \pi -\mu. 
\end{eqnarray}
Imposing $B_2\ge 1+\ln\pi$ at this level gives the simple sufficient condition $\mu\le 1/2$. 
In the next subsection we show how to obtain systematically improved purity-only sufficient conditions without discarding the higher-order terms.

\subsection{A systematic purity-only relaxation for all $n\ge 2$}\label{sec:improved}

To convert $B_n$ into explicit conditions depending only on $\mu$, it is convenient to interpret the phase-space integral as an expectation value with respect to the density $W$. 
Define the random variable
\begin{eqnarray}
X \equiv \pi W(q,p),
\end{eqnarray}
where $(q,p)$ is distributed according to $W(q,p)$. Then $0\le X\le 1$ almost everywhere by the constraint \eqref{eq: Cond-pointwise}, and Eq. \eqref{eq:lowerBoundn} then becomes
\begin{eqnarray}
\nonumber B_n
&=& \ln\pi\\
&&+\int dq\,dp\, W(q,p)\sum_{k=1}^n \frac{(-1)^k}{k}\big(\pi W(q,p)-1\big)^k \nonumber\\
&=& \ln\pi + \sum_{k=1}^n \frac{1}{k}\int dq\,dp\, W(q,p)\big(1-\pi W(q,p)\big)^k \nonumber\\
&=& \ln\pi + \sum_{k=1}^n \frac{1}{k}\,\mathbb{E}\big[(1-X)^k\big],
\label{eq:Bn_positive_terms}
\end{eqnarray}
where we used 
$$(\pi W-1)^k=(X-1)^k=(-1)^k(1-X)^k$$
so the alternating signs cancel. Moreover, we have 
\begin{align}
    \mathbb{E}[X]
    &= \int dq\,dp\, W(q,p)\,\pi W(q,p) \nonumber\\
    &= \pi\int dq\,dp\, W^2(q,p) \nonumber\\
    &= \frac{\mu}{2}.
    \label{eq:EX_purity}
\end{align}
For each integer $k\ge 2$, the function $(1-x)^k$ is convex on $[0,1]$, hence Jensen's inequality \cite{HardyLittlewoodPolya} implies for $k\ge 2$ 
\begin{eqnarray}
\mathbb{E}\big[(1-X)^k\big] \ge \big(1-\mathbb{E}[X]\big)^k
= \left(1-\frac{\mu}{2}\right)^k,
\label{eq:Jensen}
\end{eqnarray}
while equality holds for $k=1$. Substituting into Eq. \eqref{eq:Bn_positive_terms} yields the following purity-only lower bound 
\begin{eqnarray}
B_n \ge \ln\pi + \sum_{k=1}^n \frac{1}{k}\left(1-\frac{\mu}{2}\right)^k,\quad n\ge 2.
\label{eq:Bn_purity_only}
\end{eqnarray}

\subsection{Explicit sufficient thresholds and the large-$n$ limit}\label{sec:thresholds}

Following the previous section, a sufficient condition for the Wigner entropy conjecture is $B_n\ge 1+\ln\pi$. By Eq. \eqref{eq:Bn_purity_only}, it suffices that
\begin{eqnarray}
\sum_{k=1}^n \frac{t^k}{k} \ge 1,\qquad t\equiv 1-\frac{\mu}{2}\in[0,1].
\label{eq:tn_condition}
\end{eqnarray}
Equivalently, for each $n\ge 2$, let $t_n\in(0,1)$ denote the unique solution of $\sum_{k=1}^n t_n^k/k = 1$ 
(uniqueness follows since $\sum_{k=1}^n t^k/k$ is strictly increasing on $[0,1]$), then
\begin{eqnarray}
\mu \le \mu_n \equiv 2(1-t_n)\quad \Longrightarrow\quad S[W]\ge 1+\ln\pi.
\end{eqnarray}

For example, when $n=2$, Eq.~\eqref{eq:tn_condition} becomes $t+t^2/2\ge 1$, which implies $t\ge \sqrt{3}-1$ and thus
\begin{eqnarray}
\mu \le 4-2\sqrt{3}\approx 0.5359,
\end{eqnarray}
which is a strictly weaker (i.e. better as it is more inclusive) restriction than $\mu\le 1/2$ obtained above by dropping $\int W^3$.

Finally, as $n\to\infty$,
\begin{eqnarray}
\sum_{k=1}^\infty \frac{t^k}{k} = -\ln(1-t).
\end{eqnarray}
So, Eq. \eqref{eq:tn_condition} becomes $-\ln(1-t)\ge 1$, i.e.\ $1-t\le e^{-1}$. Since $1-t=\mu/2$, we obtain the simple sufficient condition when $n\to\infty$: 
\begin{eqnarray}
\mu \le \frac{2}{e}\approx 0.7358 \quad\Longrightarrow\quad S[W]\ge 1+\ln\pi.
\end{eqnarray}

\subsection{Optimality of the purity-only threshold under constraints \eqref{eq: Cond-nonneg}--\eqref{eq: Cond-pointwise}}
\label{sec:purity_optimality}

The sufficient condition $\mu\le 2/e$ obtained in the large-$n$ limit is in fact sharp under the minimal constraints \eqref{eq: Cond-nonneg}--\eqref{eq: Cond-pointwise}: no strictly larger purity threshold can hold without additional assumptions. This can be seen by an explicit extremal construction. 

\begin{proposition}[Sharpness of the purity-only bound]
\label{prop:sharp_purity_only}
Fix any $\mu\in(0,2]$. There exists a nonnegative, normalized function $W$ obeying
\eqref{eq: Cond-nonneg}--\eqref{eq: Cond-pointwise} with purity
$\mu=2\pi\int dq\,dp\,W(q,p)^2$ and Wigner entropy
\begin{align}
S[W]=\ln\!\left(\frac{2\pi}{\mu}\right).
\end{align}
Consequently, if $\mu>2/e$, then there exists such a $W$ for which $S[W]<1+\ln\pi$.
In particular, no sufficient condition of the form ''$\mu\le \mu_\star \Rightarrow S[W]\ge 1+\ln\pi$''
can hold with $\mu_\star>2/e$ unless further constraints beyond \eqref{eq: Cond-nonneg}--\eqref{eq: Cond-pointwise} are imposed.
\end{proposition}

\begin{proof}
Let $R:=\sqrt{2/\mu}$ and define the ''flat-top'' density on phase space
\begin{align}
W_{\mathrm{flat}}(q,p):=
\begin{cases}
\displaystyle \frac{\mu}{2\pi}, & q^2+p^2\le R^2,\\[4pt]
0, & \text{otherwise}.
\end{cases}
\end{align}
Then constraint \eqref{eq: Cond-nonneg} is immediate. Moreover, since the disk of radius $R$ has area
$\pi R^2 = 2\pi/\mu$, we have normalization
\begin{align}
\int dq\,dp\, W_{\mathrm{flat}}(q,p)
=\frac{\mu}{2\pi}\frac{2\pi}{\mu}=1,
\end{align}
so \eqref{eq: Cond-regularization} holds. The pointwise constraint \eqref{eq: Cond-pointwise} follows
because
\begin{align}
\pi W_{\mathrm{flat}}(q,p)=\frac{\mu}{2}\le 1
\qquad (\mu\in(0,2]).
\end{align}
Its purity is
\begin{align}
2\pi \int dq\,dp\, W_{\mathrm{flat}}(q,p)^2
=2\pi \left(\frac{\mu}{2\pi}\right)^2 \frac{2\pi}{\mu}
=\mu.
\end{align}
Finally, its entropy can be computed explicitly:
\begin{align}
S[W_{\mathrm{flat}}]
&=-\int dq\,dp\, W_{\mathrm{flat}}(q,p)\,\ln W_{\mathrm{flat}}(q,p)\nonumber\\
&=-\ln\!\left(\frac{\mu}{2\pi}\right)\int dq\,dp\, W_{\mathrm{flat}}(q,p)\nonumber\\
&=\ln\!\left(\frac{2\pi}{\mu}\right).
\end{align}
Therefore $S[W_{\mathrm{flat}}]\ge 1+\ln\pi$ is equivalent to
$\ln(2/\mu)\ge 1$, i.e.\ $\mu\le 2/e$. If $\mu>2/e$, the same construction yields
$S[W_{\mathrm{flat}}]<1+\ln\pi$, proving the claimed sharpness.
\end{proof}

\begin{remark}\label{rem:flat_nonphysical}
The flat-top extremizer also illustrates the gap between the relaxed and
physical classes. For $\mu=1$,
\begin{align}
W_{\mathrm{flat}}(q,p)
=\frac{1}{2\pi}\mathbf{1}_{\{q^2+p^2\le 2\}}
\end{align}
is nonnegative, normalized, and satisfies the pointwise Wigner bound. Let $\rho_{\mathrm{flat}}$ denote the operator obtained from $W_{\mathrm{flat}}$ by the inverse Weyl transform. Using the Wigner function $W_2$ of the Fock state $\ket{2}$, 
\begin{align}
\bra{2}\rho_{\mathrm{flat}}\ket{2}
&=2\pi\int dq\,dp\,W_{\mathrm{flat}}(q,p)W_2(q,p) \nonumber\\
&=1-9e^{-2}<0. 
\end{align}
Hence, its inverse Weyl transform is not positive semidefinite and
$W_{\mathrm{flat}}$ is not quantum realizable.
\end{remark}

Proposition~\ref{prop:sharp_purity_only} explains why any attempt to push purity-based guarantees
toward the extremal regime $\mu\to 1$ must exploit additional structure beyond \eqref{eq: Cond-nonneg}--\eqref{eq: Cond-pointwise}, such as nontrivial constraints on higher moments (e.g.\ $\int W^3$) or quantum-realizability conditions for physical Wigner functions.

\subsection{A variance-refined purity criterion at order $n=2$}
\label{sec:var_refined}

Proposition~\ref{prop:sharp_purity_only} shows that under the minimal constraints
\eqref{eq: Cond-nonneg}--\eqref{eq: Cond-pointwise}, no purely purity-based sufficient condition can improve upon $\mu\le 2/e$. Nevertheless, the hierarchy $B_n$ can yield strictly stronger criteria once one incorporates any additional information that excludes the flat-top extremizer. A convenient abstract way to quantify such information is through fluctuations of the random variable $X=\pi W(q,p)$ under the measure $W$. 

Recall that $\mathbb{E}[X]=\mu/2$ [cf.~Eq.~\eqref{eq:EX_purity}]. We also have 
\begin{align}
\mathrm{Var}(X):=\mathbb{E}\!\left[(X-\mathbb{E}[X])^2\right].
\end{align}

\begin{proposition}[Variance-refined $n=2$ bound]
\label{prop:var_refined_B2}
Under \eqref{eq: Cond-nonneg}--\eqref{eq: Cond-pointwise}, the $n=2$ bound satisfies
\begin{align}
B_2 \ge \ln\pi + t + \frac{t^2}{2} + \frac{\mathrm{Var}(X)}{2},
\quad t:=1-\frac{\mu}{2}.
\label{eq:B2_var_refined}
\end{align}
In particular, a sufficient condition for the Wigner-entropy conjecture is
\begin{align}
t+\frac{t^2}{2}+\frac{\mathrm{Var}(X)}{2}\ge 1
\quad\Longrightarrow\quad
S[W]\ge 1+\ln\pi.
\label{eq:var_suff_condition}
\end{align}
Equivalently, if $\mathrm{Var}(X)\ge v_0$ for some $v_0\in[0,1]$, then
\begin{align}
\mu \le 4-2\sqrt{\,3-v_0\,}
\quad\Longrightarrow\quad
S[W]\ge 1+\ln\pi.
\label{eq:var_threshold}
\end{align}
\end{proposition}

\begin{proof}
Starting from Eq.~\eqref{eq:Bn_positive_terms} at $n=2$ we have
\begin{align}
B_2
&=\ln\pi+\mathbb{E}[1-X]+\frac{1}{2}\mathbb{E}[(1-X)^2].
\end{align}
Let $t:=\mathbb{E}[1-X]=1-\mathbb{E}[X]=1-\mu/2$. Expanding the second moment yields the exact
identity
\begin{align}
\mathbb{E}[(1-X)^2]
&=\mathbb{E}\!\left[(1-\mathbb{E}[X]-(X-\mathbb{E}[X]))^2\right] \nonumber\\
&=t^2+\mathrm{Var}(X).
\end{align}
Substituting into the expression for $B_2$ gives \eqref{eq:B2_var_refined}.
Imposing $B_2\ge 1+\ln\pi$ and canceling $\ln\pi$ yields \eqref{eq:var_suff_condition}.
Finally, \eqref{eq:var_suff_condition} is equivalent to
$t^2+2t+\mathrm{Var}(X)-2\ge 0$; solving this quadratic with $t\in[0,1]$ gives
$t\ge \sqrt{3-\mathrm{Var}(X)}-1$, which implies \eqref{eq:var_threshold} whenever
$\mathrm{Var}(X)\ge v_0$.
\end{proof}

\begin{remark}
When $\mathrm{Var}(X)=0$, the random variable $X$ is almost surely constant under $W$, which
corresponds precisely to the flat-top extremizer in Proposition~\ref{prop:sharp_purity_only}.
In this case \eqref{eq:var_threshold} reduces to $\mu\le 4-2\sqrt{3}$, matching the $n=2$
purity-only relaxation. Any additional constraint that enforces a nontrivial lower bound
$\mathrm{Var}(X)\ge v_0>0$ therefore yields a strictly less restrictive purity threshold at order $n=2$.
\end{remark}

\section{Conclusion}\label{sec: conclusion}

We developed an analytically tractable framework for establishing lower bounds on the Wigner entropy for Wigner non-negative states. Starting from the minimal constraints \eqref{eq: Cond-nonneg}--\eqref{eq: Cond-pointwise}, we constructed a hierarchy of bounds $B_n$ (Sec.~\ref{sec: bounds}) by combining a truncated series lower bound for $-\ln x$ with moment identities of the Wigner function. This yields a systematic route to sufficient conditions ensuring the conjectured inequality $S[W]\ge 1+\ln\pi$.

A central outcome is a family of explicit purity-based tests. By interpreting $X=\pi W(q,p)$ as a bounded random variable with respect to the probability density $W$ and applying Jensen's inequality, we derived a purity-only relaxation of $B_n$ for all $n\ge 2$ (Sec.~\ref{sec:improved}). This leads to closed-form thresholds $\mu\le \mu_n$ guaranteeing $S[W]\ge 1+\ln\pi$ (Sec.~\ref{sec:thresholds}), including $\mu\le 4-2\sqrt{3}$ and $\mu\le 2/e$. Therefore, our results render all numerical attempts to test the conjecture for states with $\mu \leq 2/e$ unnecessary. 
Rewording that, we could say that, until a full proof of the conjecture, 
numerical tests may focus on high-purity Wigner-non-negative states in the remaining mixed-state regime $2/e<\mu<1$, while the pure-state endpoint $\mu=1$ is already characterized by Hudson's theorem
\cite{Hudson1974}. 

Beyond providing fast analytical checks applicable to broad families of Wigner non-negative phase-space densities, our results also clarify the limitations of relying on purity alone under \eqref{eq: Cond-nonneg}--\eqref{eq: Cond-pointwise}. 
In particular, if one aims to extend this purity-based approach to the remaining mixed-state regime $2/e<\mu<1$, one may need to incorporate additional quantum-realizability constraints---that is, conditions ensuring that $W$ comes from a positive semidefinite density operator---beyond our minimal assumptions. 
At a more abstract level, our framework already indicates concrete ``interfaces'' for such refinements: for example, at order $n=2$ the bound improves as soon as one can exclude the flat-top extremizer by enforcing a lower bound on $\mathrm{Var}(X)$, or, more generally, by controlling higher moments of $W$. 
A natural direction is to combine the present hierarchy with physicality constraints expressed at the level of characteristic functions or positive-type conditions, or with structural constraints available for special families (e.g., passive states or majorization-based classes). 
Finally, while we focused on the single-mode setting for notational simplicity, the underlying construction extends to $n$-mode phase space with the corresponding $(2\pi)^n$ normalizations, and it would be interesting to investigate how the resulting thresholds compare with the best known results for multimode beam-splitter constructions.

 \begin{acknowledgements}
 The authors are thankful to Zacharie Van Herstraeten (INRIA) and Nicolas Cerf (ULB) for useful discussions. After the completion of this work we were made aware that similar unpublished results have been derived in \cite{UlysseUnpub,NicolasUnpub}.

\end{acknowledgements}

\bibliography{refs_update.bib}

\end{document}